\documentclass[acmsmall,nonacm,screen,authorversion]{acmart}

\usepackage{amsmath}
\usepackage{amsfonts}
\usepackage[linesnumbered]{algorithm2e}
\usepackage{graphicx}
\usepackage{textcomp}
\usepackage{acronym}
\usepackage{subfig}
\usepackage{pgf}

\usepackage{listings}
\usepackage{siunitx}
\usepackage{bbm}

\usepackage{acronym}

\acrodef{CDF}{Cumulative Distribution Function}
\acrodef{CMF}{Cumulative Mass Function}
\acrodef{PDF}{Probability Density Function}
\acrodef{PRNG}{Pseudo-Random Number Generator}
\acrodef{AWS}{Amazon Web Services}
\acrodef{GPU}{Graphical Processing Unit}
\acrodef{AES}{Advanced Encryption Standard}

\graphicspath{{./figures/}}

\def\ie{\textit{i.e.}}
\def\eg{\textit{e.g.}}

\def\Prob{\mathbb{P}}
\def\U{\mathcal{U}}
\def\N{\mathcal{N}}
\def\bbU{\mathbb{U}}

\DeclareMathOperator{\Lap}{Lap}
\DeclareMathOperator{\Exp}{Exp}

\newtheorem{assumption}{Assumption}

\begin{document}

\title{Secure Random Sampling in Differential Privacy}

\author{Naoise Holohan}
\email{naoise@ibm.com}
\author{Stefano Braghin}
\email{stefanob@ie.ibm.com}
\affiliation{%
\institution{IBM Research Europe -- Ireland}
\city{Dublin}
\country{Ireland}}
\authorsaddresses{}

\acmDOI{}
\acmYear{2021}
\acmMonth{5}


\begin{abstract}
Differential privacy is among the most prominent techniques for data anonymisation, and it is the first one proposed with rigorous mathematical guarantees.
Previous work demonstrated that concrete implementations of differential privacy mechanisms are vulnerable to statistical attacks.
This vulnerability is caused by the approximation of real values to floating point numbers.
This paper presents a practical solution to the finite-precision floating point vulnerability, where the inverse transform sampling of the Laplace distribution can itself be inverted, thus enabling an attack where the original value can be retrieved with non-negligible advantage.

The proposed solution has the advantages of being generalisable to any infinitely divisible probability distribution, and of simple implementation in modern architectures.
Finally, the solution has been designed to make side channel attack infeasible, because of inherently exponential, in the size of the domain, brute force attacks.
\end{abstract}

\hypersetup{
     citecolor    = ForestGreen,
 	 linkcolor    = RoyalBlue,
 	 urlcolor     = RedViolet
}

\maketitle

\section{Introduction}
The generation of random samples from probability distributions is a well-studied field, and is an ever-present feature in programming languages and random number generators alike.
Random sampling is especially important in the field of differential privacy, where specially-calibrated random noise is used to protect published data from unwanted inference.
However, while the definition and analysis of differential privacy is typically viewed in the theoretical, real-valued world of mathematics, translating its rigorous guarantees to the floating-point world of computers requires special attention.

Work by Mironov in 2012~\cite{Mir12} was the first to point to significant vulnerabilities in floating point implementations of differential privacy, which allow for catastrophic destruction of the much-vaunted privacy guarantees.
Using na\"ive sampling of the Laplace distribution, Mironov was able to exploit holes in the output space to reconstruct the original input from a finite list of candidate inputs.
Being able to reconstruct a single value with certainty is a blatant breach of differential privacy, which should guarantee uncertainty on any query.
In this paper, we further extend Mironov's approach to present a novel attack on the Gaussian mechanism (Section~\ref{sc:bg:gauss}).

It is well-known that the cardinality of the real line $\mathbb{R}$ is the same as that of the unit interval on the real line, $[0, 1]$, which allows for precise sampling using the inverse transform method.
For floating-point numbers however, there are many more numbers on the real line than there are in the unit interval, which results in the holes that Mironov used in the attack.
Typical defences against this vulnerability eliminate these holes by limiting the output space, using such techniques as the ``snapping mechanism''~\cite{Mir12}, and sampling from the discrete analogue of the distribution~\cite{Goo20,CKS20}, but these require much more complex code to implement.

In this paper, we seek an alternative defence -- that of computational complexity.
The attack presented by Mironov relies on an attacker inverting the sampling procedure to test the feasibility of candidate inputs.
Although random samples generated from a single random number (\eg, the Laplace distribution) are vulnerable to this attack, samples generated by two or more random numbers (\eg, the Gaussian distribution) are less vulnerable.
By generating samples from multiple random numbers, we can render the attack sufficiently costly so as to be impractical on today's computers.

We demonstrate how this approach can be implemented with ease for the Laplace and Gaussian distributions (both popular in differential privacy), and a large family of probability distributions.
The approach can be extended to systems of reduced precision (\eg, single- and half- precision floating point), and allows for any desired level of complexity to be achieved.

Given that data is commonly stored as floating point numbers, we believe there is value in developing tools that reflect this reality.
Equally, we are cognisant of the need to develop simple solutions to these complex tasks, to allow for easy implementation, adoption, generalisation and understanding of the techniques.
As described in~\cite{Dev86}:
\begin{quote}
Fast programs are seldom short, and short programs are likely to be slow.
But it is also true that long programs are often not elegant and more error-prone.
Short smooth programs survive longer and are understood by a larger audience.
\end{quote}

The rest of the paper is organised as follows.
Section~\ref{sc:bg} describe the required background information.
The guiding principles to this paper are then given in Section~\ref{sc:gen}, with the specifics presented in Section~\ref{sc:main}.
Example implementations are given in Section~\ref{sc:samp} and the protection provided is explored in Section~\ref{sec:guarantees}
After reviewing the state of the art in Section~\ref{sec:related-work}, the contribution is summarised in Section~\ref{sec:conclusion}.

\section{Background}\label{sc:bg}
This section will give a brief outline of the vulnerability presented in~\cite{Mir12}, alongside an example attack implementation and its extension to attacking Gaussian sampling, and the current state-of-the-art in mitigating against the attack.

\subsection{Floating point numbers}
\label{sec:background:floating-point-numbers}

To begin, we give a very brief outline of floating point numbers.
Floating point numbers take their inspiration from scientific notation to store a wide range of numbers in binary format.
Double-precision floating points (also known as \emph{doubles}) occupy \num{64} bits of storage, comprising \num{1} bit for the sign, \num{11} bits for the exponent and the remaining \num{52} bits for the fraction or mantissa.\footnote{\label{foot:ieee754}ANSI\slash IEEE Std 754--2019
\url{http://754r.ucbtest.org}}
The corresponding real number for a given double is given by
$$(-1)^{\text{sign}} (1.b_{51} b_{50} \dots b_{0})_2 \times 2^{e - 1023},$$
where $b_0, \dots, b_{51}\in \{0, 1\}$ are the bits of the mantissa, and $e \in \mathbb{N}$ is the non-negative integer exponent.
This format allows for the representation of numbers between $10^{-308}$ and $10^{308}$, with varying degrees of granularity.

The IEEE standard$^{\ref{foot:ieee754}}$ gives an algorithm for addition, subtraction, multiplication, division, and square root and requires that implementations produce the same result as that algorithm.

The standard specifies that floating-point numbers are represented as base \num{2} fractions.
For example, the value $0.001_2$ represents the decimal value $\frac{0}{2} + \frac{0}{4} + \frac{1}{8}$.
Nonetheless, not all decimal fractions can be represented exactly as binary fractions.
Thus, the decimal floating-point numbers are approximated to the binary floating point.
For example, in base-\num{2}, $\frac{1}{10}$ is represented as the infinitely-repeating value $0.0001100110011\ldots$

The standard provides information about the error introduced by both the representation and operations that can be performed.
For example, the addition (and subtraction) operation can be performed by converting the operands to the same exponent, and then summing (or subtracting) the mantissa.
On the other hand, multiplications can be computed by multiplying the mantissa of both operands and adding the exponents.

These operation introduce a rounding error, as the values need to be converted and as multiplications are performed.
Moreover, the standard specifies various rounding modes, including
round to nearest (where ties round to the nearest even digit in the required position, the most commonly used),
round away from zero,
round up (toward $+\infty$),
round down (toward $-\infty$), and
round toward zero (\ie, truncation).

The amount of error can be quantified~\cite{goldberg1991every} to a \emph{machine epsilon} (generally denoted as $E_{\mathtt{mach}}$), that depends on the actual precision used in the representation.
For example, single precision floating point (\num{32} bits) have a machine epsilon equal to $2^{-23} \approx 1.19 \times 10^{-7}$, while double precision floating point numbers (\num{64} bits) have a  machine epsilon equal to $2^{-52} \approx 2.22 \times 10^{-16}$.

\subsection{Random number sampling}
When sampling random numbers from a probability distribution, the common starting point is to take a (many) random number(s) from the unit interval and transform them to sample from the given distribution.
Although the unit interval can be sampled with increasing granularity closer to zero, the typical practice is to sample from a multiple of a small number.
For example, the Python programming language returns a multiple of $2^{-53}$ when sampling from the unit interval.

Inverse transform sampling is a popular method to sample from certain probability distributions, by taking the inverse of the \ac{CDF} or \ac{CMF}.
For example, given a \ac{CDF} $F(x): \mathbb{R} \to [0, 1]$ from which we want to sample, and whose inverse $F^{-1}(x): [0, 1] \to \mathbb{R}$ is known, then given $U \sim \U(0, 1)$, we have
$$\Prob(F^{-1}(U) \le x) = \Prob(U \le F(x)) = F(x),$$
for $x \in \mathbb{R}$, noting that $\Prob(U \le u) = u$ for a uniform variate $U \sim \U(0, 1)$ and $u \in [0, 1]$.

\subsection{Mironov attack}\label{sc:bg:attack}
Noting that there are only $2^{53}$ possible uniform variates to be generated and $2^{64}$ floating point numbers, there is an immediate difficulty in covering the output space with the image of a transformation.
Without additional randomness, it is impossible for the inverse transform sampling method to fully cover the real line.
This was demonstrated in~\cite{Mir12}, where the author confirmed that gaps appear in the output space between the outputs of consecutive uniform variates.
Mironov formulated an attack that could successfully reconstruct an entire database, with certainty, under specific conditions.

The \ac{CDF} of the standard Laplace distribution is given by
\begin{equation}
F_{\Lap}(x) = \begin{cases}
\frac{1}{2} e^{x} & \text{if } x \le 0,\\
1- \frac{1}{2} e^{-x} & \text{if } x > 0,
\end{cases}\label{eq:cdf:lap}
\end{equation}
with which we can calculate its inverse:
\begin{equation}\label{eq:inv:lap}
F_{\Lap}^{-1}(u) = (-1)^{\lfloor u \rceil_1} \log(1 - 2 \left|u-0.5\right|),
\end{equation}
where $\lfloor \cdot \rceil_k$ denotes rounding to the nearest multiple of $k \in \mathbb{
R}$.
A Laplace variate can then be generated using $F_{\Lap}^{-1}(U) \sim \Lap(0, 1)$, by first sampling a uniform variate $U \sim U(0, 1)$.
We present an example algorithm for attacking a single value in a dataset in Algorithm~\ref{alg:mironov}.

\begin{algorithm}
\caption{Example algorithm implementing the Mironov attack}\label{alg:mironov}
\KwIn{Attack target $v$, DP query $Q$ using the Laplace mechanism with implementation (\ref{eq:inv:lap}), finite candidate set $\mathcal{C}$}
\KwOut{Attack result $c \in \mathcal{C}$}
\BlankLine

\While{$|\mathcal{C}| > 1$}{
$q = Q(v)$\label{alg:q}

\For{$c \in \mathcal{C}$}{
$u = \lfloor F_{\Lap}(q - c) \rceil_{2^{-p}}$\label{alg:u}

\If{$F_{\Lap}^{-1}(u) + c \neq q$\label{alg:F}}
{remove $c$ from $\mathcal{C}$}
}
}
\KwRet{remaining element $c \in \mathcal{C}$}
\BlankLine
\end{algorithm}

Critical to this attack is the ability to compute the inverse of the sampling procedure (\ie, the \ac{CDF}, given in (\ref{eq:inv:lap}), in the case of inverse transform sampling).
The success of this attack is therefore independent of the precision of the uniform variate, and can be executed whenever the sampling function can be inverted.

\subsection{Gaussian attack}\label{sc:bg:gauss}
We now show an extension of the Mironov attack to the Gaussian distribution, one which we believe to be novel.
This attack can be performed on Gaussian variates sampled using the popular Box-Muller transform~\cite{BM58}.
Given two uniform variates $U_1, U_2 \sim \U(0, 1)$, the Box-Muller transform returns two independent Gaussian samples $N_1, N_2 \sim N(0, 1)$ as follows:
\begin{subequations}
\begin{align}
N_1 = \sqrt{-2 \log(1 - U_1)} \cos(2\pi U_2), \label{eq:boxmuller:cos}\\
N_2 = \sqrt{-2 \log(1 - U_1)} \sin(2\pi U_2).
\end{align}
\end{subequations}
Knowing both $N_1$ and $N_2$, we can recover $U_1$ and $U_2$ as follows:
\begin{subequations}
\label{eq:bg:gauss}
\begin{align}
U_1 &= 1 - e^{-\frac{N_1^2 + N_2^2}{2}},\\
U_2 &= \frac{1}{2\pi}\left(\arctan\left(\frac{N_2}{N_1}\right) + \pi \mathbbm{1}_{\{N_1 < 0\}}\right).
\end{align}
\end{subequations}
The same inversion technique as described in Algorithm~\ref{alg:mironov} can be used to eliminate candidates and determine the true input, noting that $q$ and $u$ in Lines~\ref{alg:q}, \ref{alg:u} and \ref{alg:F} will be vectors of two values each.

Getting two values from a mechanism utilising the Box-Muller transform can be done by executing the same query twice in succession.
Ensuring both values are from the same Box-Muller operation, and determining which is associated with the $\cos$ and which is associated with the $\sin$ can be done by examining the source code and executing a simple timing attack.
It is a common implementation to return the first variate to the user, while caching the second variate to be returned the next time the function is called~\cite{TLL07}.
The calculation overhead in the first step can be measured to determine the phasing to implement the attack.

This attack can be mitigated against by discarding one of the variates, or by using both variates in an output of $\frac{1}{\sqrt{2}}(N_1 + N_2)$.

\subsection{Existing defences}
To mitigate against the attack, Mironov proposed the \emph{snapping mechanism} to sample a noisy output with Laplace-like additive noise.
The snapping mechanism involves snapping the noisy output to the nearest factor of the scale of the noise, resulting in a significant reduction in granularity of the output.
For example, given a privacy budget $\epsilon= 2^{-5}$, the only outputs will be multiples of $2^5=32$.
Alongside the reduction in granularity, the snapping mechanism is also cumbersome to implement, requiring a custom approach to sampling from the unit interval, as well as clipping and rounding operations.

Another approach is to sample from the discrete analogue of the desired distribution, thereby giving more control to the user on the discretisation of the output~\cite{Goo20,CKS20}.
\cite{Goo20} in particular offers solutions for the Laplace and Gaussian distributions, but requires complex sampling procedures.
It is not clear if it is simple, or possible, to adapt their solution to other probability distributions.

While there are many algorithms to sample from the discrete Gaussian distribution~\cite{CKS20,Kar14}, all common random number generator packages produce floating-point Gaussian variates, making the implementation of such solutions more complex.

\section{General principles}\label{sc:gen}
In this section we present the general principles of the proposed solution.
We first make the following assumptions in analysing the protections provided by our approach:
\begin{assumption}\label{as:1}
Mechanism source code is public.
\end{assumption}

\begin{assumption}\label{as:2}
Uniform variates are generated using a cryptogra\-phically-secure pseudorandom number generator (CSPRNG).
\end{assumption}

\begin{assumption}\label{as:3}
Uniform variates are generated as multiples of $2^{-p}$ for some fixed $p \in \mathbb{N}_{>0}$.

Hence, for $U \sim \U(0, 1)$, we have
\begin{equation}
U \in [0, 1) \cap \left(2^{-p} \mathbb{N}\right).
\end{equation}
\end{assumption}

Assumption~\ref{as:1} ensures that we cannot achieve security of random number generation through obscurity, and that the protections are an inherent characteristic of the system.
Assumption~\ref{as:2} ensures that an attacker cannot predict forthcoming uniform variates, as can be done with some standard \acp{PRNG} (\eg, the Mersenne-Twister \ac{PRNG}).
Assumption~\ref{as:3} simplifies some of the analysis of this paper, and aligns with random number generation in many programming languages.

For clarity, we define the set of uniform variates as follows for a given precision $p$.

\begin{definition}[Uniform variates]\label{df:uvar}
Given $p \in \mathbb{N}$, we define
$$\bbU = [0, 1) \cap (2^{-p} \mathbb{N}),$$
the set of all multiples of $2^{-p}$ in the half-open unit interval $[0, 1)$.

This denotes the set of uniform variates in many programming languages.
For example, in the Python programming language, $p = 53$.
\end{definition}

Although the holes in the output space of the random number generation provide the key to the attack as shown in Section~\ref{sc:bg:attack}, it relies on the invertibility of the sampling procedure to execute efficiently.
If the sampling procedure can be made non-injective (and consequently, non-invertible), then any sampled variate could be the result of more than one combination of uniform variate(s).
Therefore, satisfying the non-equality condition of Line~\ref{alg:F} in Algorithm~\ref{alg:mironov} would require searching through all possible uniform variate combinations that approximately give equality in Line~\ref{alg:F} (up to floating point rounding errors).

We can flesh out the details of this using the Gaussian distribution as an example.

\begin{example}[Gaussian distribution]\label{eg:gen:gauss}
Consider samples from the Gaussian distribution, as given in Section~\ref{sc:bg:gauss}.
Although it is possible to reconstruct $U_1$ and $U_2$ knowing both outputs from the Box-Muller transform, we can show that knowing a single Box-Muller output is not sufficient to execute the Mironov attack in a single step.
Suppose $N_1$ is known to have originated from (\ref{eq:boxmuller:cos}), and we wish to determine $U_1$ and $U_2$ from which it came.
We can write
\begin{equation}\label{eq:gen:gauss}
U_1(U_2) = 1 - e^{-\frac{N_1^2}{2 \cos^2(2 \pi U_2)}},
\end{equation}
allowing us to solve for $U_1$ given $U_2$.
We can similarly write $U_2$ as a function of $U_1$.

As shown in Algorithm~\ref{alg:mironov}, we need only find a single pair $U_1, U_2 \in \bbU^2$ that gives equality in Line~\ref{alg:F} in order to retain that candidate as feasible.
On the real line, any pair $U_1, U_2 \in [0, 1)$ satisfying (\ref{eq:gen:gauss}) will give equality for Line~\ref{alg:F}.
However, due to floating point rounding errors and the limited precision of values $U \in \bbU$, we need to search through $\bbU^2$ to find values that satisfy the equality down to the bit of least precision.
Figure~\ref{fig:eg:1} shows the level curves of (\ref{eq:gen:gauss}), illustrating the paths through the grid of $(U_1, U_2) \in \bbU^2$ that would be traversed in a typical execution of the attack.

\begin{figure}[t]
\begin{center}
\includegraphics[width=0.65\columnwidth]{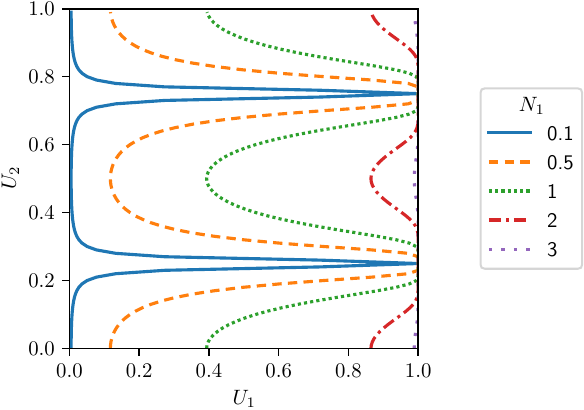}
\end{center}
  \caption{\label{fig:eg:1} Level curves of (\ref{eq:gen:gauss}) for given values of $N_1$.}
\end{figure}
\end{example}

This example prompts two questions:
\begin{enumerate}
\item Can we extend the protection offered to Gaussian variates, beyond that of searching through $2^p$ uniform variates?
\item Can this approach be adopted for other probability distributions? Particularly, ones with injective sampling procedures (\eg, the Laplace distribution)?
\end{enumerate}

If we can use our choice of $n \in \mathbb{N}$ uniform variates to generate a single sample from a given distribution, then a brute-force attack (similar to Example~\ref{eg:gen:gauss}) would require on the order of $2^{p(n-1)}$ `checks'.
In essence, it would require fixing $(n-1)$ of the variates, and running an inversion attack (\eg, using (\ref{eq:gen:gauss}) in the case of the Gaussian distribution) on the $n$th variate.
Importantly, for a linear increase in the time to sample a variate, the cost of executing the attack will increase exponentially.

Additionally, this approach can be extended to systems of reduced precision.
Single- and half- precision systems use only \num{32} and \num{16} bits respectively to represent a floating point number, compared to the \num{64} bits allocated for double-precision floating point numbers.
For example, \acp{GPU} typically use lower-precision floating point numbers in their calculations, making attacks such as Example~\ref{eg:gen:gauss} more feasible.

In implementing such an approach, we can take some inspiration from the subset sum problem in computer science.
The subset sum problem asks, given a multiset $A$ of positive integers ($a \in \mathbb{N}$ for each $a \in A$), and a target $T \in \mathbb{Z}$, is there a subset $A^\prime \subseteq A$ such that $\sum_{a \in A^\prime} a = T$?
It is known that the subset sum problem is NP-complete~\cite{blackbook}, making it intractable to solve when the set $A$ is large.

Luckily, there is a branch of probability theory that we can exploit in this context -- infinite divisibility.

\section{Divisibility of probability distributions}\label{sc:main}
Our approach relies on simple and effective ways to generate random numbers from a given distribution using many uniform variates.
As explained in Section~\ref{sc:gen}, simply increasing the number of random bits used to generate a random sample from a distribution is not sufficient to protect against inversion attacks.
Instead, we seek to increase the number of uniform variates used to generate a single random sample.

A probability distribution $P$ is \emph{divisible}, if there exists a distribution $Q$ such that, given independent $X_1, X_2 \sim Q$, we have $X_1 + X_2 \sim P$.
$P$ is \emph{infinitely divisible} if, for any $n \in \mathbb{N}$, there exists a distribution $R_n$ such that independent $X_i \sim R_n$ for each $i \in [n]$ satisfy $\sum_{i=1}^n X_i \sim P$.

Probability distributions that are infinitely divisible therefore allow the greatest flexibility in increasing the complexity of this defence.
Fortunately, the two most popular distributions used in differential privacy are infinitely divisible.

\subsection{Preliminaries}
For every $i \in \mathbb{N}$, we adopt the following notation for these common distributions:
\begin{itemize}
\item $N_i \sim \mathcal{N}(0, 1)$ is a collection of independent (standard) Gaussian random variables, and,
\item $U_i \sim \U(0, 1)$ is a collection of independent uniform random variables on the unit interval.
\end{itemize}
Formulations of the probability distributions used in this section are given in Appendix~\ref{sc:ap:pdf},

For any $n \in \mathbb{N}$, we let $[n] = [1, n] \cap \mathbb{N}$ denote the positive integers up to and including $n$.

We note that when sampling uniformly from the unit interval, random number generators typically sample from $[0, 1)$ (\ie, excluding \num{1}).
This range is undesirable when computing $\log(U_1)$, and instead we use $\log(1 - U_1)$, since $1 - U_1 \in (0, 1]$ lies within the domain of logarithms.

Finally, we limit our analysis to standard distributions, those with zero location and unit scale parameters.
These standard distributions can be scaled and translated as appropriate to get a distribution of a specific mean and variance (see Appendix~\ref{sc:ap:pdf}).

\subsection{Gaussian distribution}\label{sc:main:gauss}
It is well-known that the Gaussian distribution is infinitely divisible.
Given that the sum of two Gaussians is Gaussian ($N_1 \pm N_2 \sim \mathcal{N}(0, 2)$), we can generate a standard Gaussian from $n > 0$ standard Gaussians:
\begin{equation}\label{eq:gau:div}
\frac{1}{\sqrt{n}} \sum_{i=1}^n N_i \sim \mathcal{N}(0, 1).
\end{equation}

This property of divisibility allows for the sampling of Gaussians with ease for any $n \in \mathbb{N}$.
Additionally, decomposing any other probability distribution as a function of Gaussians will allow for the subsequent decomposition into any finite number of Gaussians.

\subsection{Laplace distribution}
The Laplace distribution is infinite divisible using the gamma distribution.

\begin{proposition}[Laplace divisibility~\cite{KKP12}]
Given an integer $n \ge 1$, and $X_i, Y_i \sim \Gamma\left(\frac{1}{n}, 1\right)$ for each $i \in [n]$, then
$$\sum_{i =1}^n\left(X_i - Y_i\right) \sim \Lap(0, 1).$$
\end{proposition}

The case of $n=1$ gives the decomposition of the Laplace distribution as the difference of two exponential distributions.

\begin{corollary}
Given $E_1, E_2 \sim \Exp(1)$,
\begin{equation}\label{eq:exp}
E_1 - E_2 \sim \Lap(0, 1).
\end{equation}
\end{corollary}

\begin{proof}
Noting that $\Gamma(1, 1) \sim \Exp(1)$ completes the proof.
\end{proof}

The special case of $n=2$ also deserves special attention.

\begin{corollary}\label{cr:lap:1}
Given independent samples of the standard Gaussian distribution, $N_1, N_2, N_3, N_4 \sim \N(0, 1)$,
\begin{equation}\label{eq:lap:sqsum}
\frac{1}{2}\left(N_1^2 - N_2^2 + N_3^2 - N_4^2\right) \sim \Lap(0, 1).
\end{equation}
\end{corollary}

\begin{proof}
The chi-squared distribution with \num{1} degree of freedom is related to the gamma distribution (given $X \sim \chi^2(1)$, then $\frac{1}{2} X \sim \Gamma\left(\frac{1}{2}, 1\right)$).
Furthermore, it is known that a squared sample from the Gaussian distribution is chi-squared ($N_1^2 \sim \chi^2(1)$).
Combining these two observations completes the proof.
\end{proof}

The Laplace distribution can be decomposed in many other ways, as given in~\cite[Table~2.3]{KKP12}, one of which is of particular interest in our context:

\begin{proposition}\label{pr:lap:1}
Given independent samples of the standard Gaussian distribution, $N_1, N_2, N_3, N_4 \sim \N(0, 1)$,
\begin{equation}\label{eq:lap:proddiff}
N_1 N_2 - N_3 N_4 \sim \Lap(0, 1).
\end{equation}
\end{proposition}

\begin{proof}
We note that $Z = N_1 N_2$ has a characteristic function $\phi_Z(t) = (1 + t^2)^{-\frac{1}{2}}$.
This gives the characteristic function for (\ref{eq:lap:proddiff}) of 
$$\phi_{N_1 N_2 - N_3 N_4}(t) = \phi_Z(t)\phi_Z(-t) = (1 + t^2)^{-1},$$
which corresponds to the standard Laplace distribution.
\end{proof}

We therefore have two methods to sample from the Laplace distribution as a composition of four standard Gaussians.
Each Gaussian can be further decomposed using (\ref{eq:gau:div}).

\section{Sampling implementations}\label{sc:samp}
There are many ways to sample from the Gaussian distribution~\cite{TLL07}.
We restrict our analysis to the Box-Muller transform~\cite{BM58}, owing to its simplicity, compactness, exactness and ubiquity among programming languages.

The polar method~\cite{Bel68} of the Box-Muller transform was developed to avoid the costly $\sin$ and $\cos$ calculations.
Other methods have been proposed to eliminate the similarly costly $\log$ and $\sqrt{\cdot}$ calculations~\cite{AD88,Bre93}.
These modifications are typically reserved for low-level languages or machine code, because the added complexity can slow computation time in higher-level languages~\cite{AD72b}.
Given our focus on high-level languages, such as Python, analysis of these variants is beyond the scope of this paper.

In this section, we offer simple implementation examples in Python, using standard libraries.
We hope these examples will be sufficient to enable the reader to port the code to their desired language.

\subsection{Gaussian sampling}
As noted in Section~\ref{sc:bg:gauss}, a na\"ive implementation of the Box-Muller transform in sampling from the Gaussian distribution presents vulnerabilities.
The question therefore arises of how best to protect Gaussian samples from the Mironov attack, and there are two immediate considerations:
\begin{enumerate}
\item Discard one of the pair of Box-Muller samples, or
\item Make use of both samples, noting that $\frac{1}{\sqrt{2}}(N_1 + N_2) \sim \N(0, 1)$.
\end{enumerate}

Both of these choices are mathematically identical (albeit with the outputs $\frac{\pi}{4}$ out of phase\footnote{Since $\cos(\theta) + \sin(\theta) = \sqrt{2} \cos\left(\theta - \frac{\pi}{4}\right)$}), but it may be a more elegant implementation using standard libraries to opt for the second option.
Extending this choice to generating Gaussian samples from multiple Box-Muller draws, using the infinite divisibility described in Section~\ref{sc:main:gauss}, we use $2n$ samples\footnote{We need only use $n$ samples for sampling procedures that do not share uniform variates between executions, for example the \texttt{normalvariate} method in Python's \texttt{random} library, which uses the Kinderman-Monahan sampling procedure~\cite{KM77}.} from the Box-Muller transform, and divide the result by $\sqrt{2n}$.
This can be implemented in Python as follows, after importing the \texttt{math} and \texttt{random} libraries:
\begin{lstlisting}
sum(gauss(0, 1) for i in range(2 * n)) / sqrt(2 * n)
\end{lstlisting}
Readers are encouraged to discard the first result from Box-Muller before sampling in this way, to ensure no carry-over of uniform variates from a previous draw, which may have already been successfully attacked.


Importantly, the time taken in the generation of a standard Gaussian variate from $n$ uniform variates increases linearly with $n$, whereas the attack complexity as described in Section~\ref{sc:gen} increases exponentially in $n$.
Therefore, a small increase in computation overhead to sample the variates gives an exponential increase in attack complexity.

\subsection{Laplace sampling}\label{sc:samp:lap}
We ignore the decomposition of the Laplace distribution into the difference of two exponential variates, and skip straight to the decomposition into Gaussians.

As such, we are interested in the forms given in Corollary~\ref{cr:lap:1} and Proposition~\ref{pr:lap:1}.
In both these cases, we can safely use two samples from Box-Muller to generate the Laplace sample, since squaring-and-summing or multiplying Box-Muller samples breaks the invertibility given in Section~\ref{sc:bg:gauss}.
Empirical analysis has shown (\ref{eq:lap:proddiff}) to be the faster option, as (\ref{eq:lap:sqsum}) requires two additional multiplication operations.
In Python, this can be implemented as follows after importing the \texttt{random} library:
\begin{lstlisting}
gauss(0, 1) * gauss(0, 1) - gauss(0, 1) * gauss(0, 1)
\end{lstlisting}

However, doing so will introduce some redundancy in the uniform variates, which presents an opportunity for those willing to write their own subroutines.

\begin{theorem}\label{th:lap}
Given independent uniform variates $U_1, U_2, U_3, U_4 \sim \U(0, 1)$,
\begin{equation}\label{eq:main:bm:lap}
\log(1 - U_1) \cos(\pi U_2) + \log(1 - U_3) \cos(\pi U_4) \sim \Lap(0, 1).
\end{equation}
\end{theorem}

\begin{proof}
Firstly with (\ref{eq:lap:sqsum}), we can write the difference of squared outputs from the Box-Muller transform as follows,
\begin{align}
N_1^2 - N_2^2 &= -2 \log(1-U_1) \left(\cos^2(2 \pi U_2) - \sin^2(2 \pi U_2)\right)\\
&=-2 \log(1-U_1) \cos(4 \pi U_2),\label{eq:main:lap:2}
\end{align}
using the double-angle formula $\cos(2 \theta) = cos^2(\theta) - \sin^2(\theta)$.

Also, since $-\cos(4 \pi U_2) \sim \cos(\pi U_2)$ when $U_2 \sim \U(0, 1)$, owing to the periodicity of $\cos$, we get the desired result.
Analysis of (\ref{eq:lap:proddiff}) gives a similar representation.
\end{proof}

In Python, this can be implemented as follows after importing the \texttt{math} and \texttt{random} libraries:
\begin{lstlisting}
log(1 - random()) * cos(pi * random()) + log(1 - random()) * cos(pi * random())
\end{lstlisting}

\textbf{Remark 1:} We observe that $U_2 \in [0, 1)$ in a floating-point environment, which results in an asymmetric output for $\cos(\pi U_2) \in (-1, 1]$.
We can eliminate this asymmetry by using the first bit of randomness of $U_2$ to sample the sign, and then using the remaining bits to sample the magnitude, giving
$$\cos(\pi U_2) \sim (-1)^{\lfloor U_2 \rceil_1} \cos\left(\pi \left(U_2 \text{ mod } \frac{1}{2}\right)\right)$$
in real number arithmetic, with an output space of $[-1, 1]\setminus \{0\}$ when $U_2 \in [0, 1)$. 
This can be implemented in Python as follows:
\begin{lstlisting}
u2 = random()
copysign(cos(pi * (u2 % 0.5)), u2 - 0.5)
\end{lstlisting}

\textbf{Remark 2:} The representation in Theorem~\ref{th:lap} brings two benefits over sampling na\"ively with Box-Muller: (i) reduced redundancy in the $\cos$ argument (using Box-Muller, the first two random bits of $U_2$ and $U_4$ are redundant), and, (ii) greater computational efficiency (empirical evidence suggests a halving in the time taken to sample).

\subsection{Choosing $n$}
In choosing the number, $n$, of uniform variates to use in sampling a single Laplace or Gaussian variate, we take inspiration from current standards in cryptography.
The \ac{AES} supports key sizes of \num{128}, \num{192} and \num{256} bits, corresponding with search spaces of sizes $2^{128}$, $2^{192}$ and $2^{256}$ respectively.
Thus, it will take $O(key~size)$ steps for a polynomial adversary to enumerate, and test, all possible keys.

In order to achieve a similar search space size for our application, assuming a precision of $p = 53$ (as in Python), we require $n = 4, 5$ and \num{6} respectively.
We consider it sufficient to use $n=4$ in most applications, corresponding to the implementation of Theorem~\ref{th:lap}.
This standard has been adopted in the implementation of the Laplace, Gaussian and other mechanisms\footnote{\url{https://github.com/IBM/differential-privacy-library/tree/main/diffprivlib/mechanisms}} in the diffprivlib open source library~\cite{HBML19}.

The effect of $n$ on the time it takes to sample from the Laplace distribution is shown in Figure~\ref{fig:lap_samp}.
Despite using four times as many uniform variates, the implementation of Theorem~\ref{th:lap} only takes twice as long as the na\"ive sampling with a single uniform variate.
Using the native \texttt{math} and \texttt{random} libraries is costly for larger $n$, owing to Python's slow \texttt{for} loops.
Numpy's C codebase~\cite{WCV11} allows for fast computation even for large $n$, and can be further leveraged to produce multiple samples in parallel, with superior per-sample computation time than na\"ive sampling.
The code for these simulations is given in Appendix~\ref{sc:ab:code}.

\begin{figure}[t]
\begin{center}
\includegraphics[width=\columnwidth]{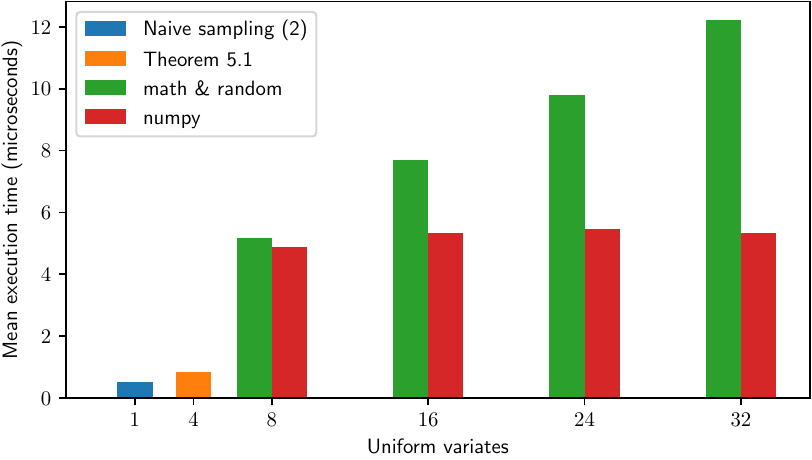}
\end{center}
  \caption{\label{fig:lap_samp} Mean execution time for sampling a single standard Laplace variate, in microseconds (\SI{}{\micro\second}), over \num{5000000} runs in Python.}
\end{figure}

\section{Gaussian attack complexity}\label{sec:guarantees}
We briefly revisit the brute force attack on the Gaussian distribution in Example~\ref{eg:gen:gauss} to highlight the robustness of the proposed approach.
Given a single Gaussian variate, we know from Example~\ref{eg:gen:gauss} that we can write $U_1$ as a function of $U_2$.
This allows us to estimate a lower bound on the number of checks required on a given variate $N_1$, since we know that $U_1, U_2 \in \bbU \subset [0, 1)$.
We stress that this is only a lower bound, as rounding errors in floating point arithmetic would typically require the checking of adjacent uniform variates.

For example, using (\ref{eq:gen:gauss}), given $U_2 \in [0, 1)$, we have
$$U_1 \in \Big[1 - e^{-\frac{N_1^2}{2}}, 1 \Big),$$
allowing us to reduce the search space to $U_1 \in \bbU \cap \big[1 - e^{-\frac{N_1^2}{2}}, 1 \big)$
We can therefore approximate the number of checks as $e^{-\frac{N_1^2}{2}} 2^p$.

Knowing that $N_1 \sim \N(0, 1)$, we can get the expected number of checks as follows:
\begin{align}
\mathbb{E}[\text{checks}] &= 2^p \int e^{-\frac{x^2}{2}} f(x) dx = 2^p C \int e^{-\frac{(\sqrt{2} x)^2}{2}} dx =2^p \frac{C}{\sqrt{2}} \int e^{-\frac{x^2}{2}} dx\\
&=\frac{2^p}{\sqrt{2}} = 2^{p - \frac{1}{2}},\label{eq:expchecks}
\end{align}
where $f(x)$ denotes the probability density function of the standard Gaussian distribution $\N(0, 1)$ and $C$ is the normalisation constant (see Appendix~\ref{sc:ap:gau}).

This demonstrates the robustness of the approach, with (\ref{eq:expchecks}) confirming that the complexity of the attack only decreases by a small constant from the theoretical limit.
While outlying values lead to smaller search spaces, their vanishingly small probabilities of being sampled result in a small aggregate impact.

\section{Related work}\label{sec:related-work}
The original attack on the Laplace mechanism, which also inspired this work, was first presented in~\cite{Mir12}.
Subsequent work on the subject in~\cite{GMP16}, where the authors analyse the situation of infinite-precision semantic at the implementation level. The authors present results arguing that in general there are violations of the differential privacy property, proposing a variation of the differential privacy definition leading to a degradation of the privacy level.

Thorough examination of randomness required for differential privacy is presented in~\cite{GF20}, where the authors analyse various techniques for random number generation, presenting a picture of strengths and limitations of the commonly used sources of randomness.

Other work in~\cite{DLM12} and~\cite{MPR09} analyse theoretical and practical limitations of the source of randomness used in differential privacy. In particular ~\cite{DLM12} presents an cryptographic requirements for real randomness, and describe the analogies existing in differential privacy.

In~\cite{GMP16} the authors show the violation of differential privacy property caused by the approximation introduced by the finite-precision representation of continuous data.
The authors also present the conditions under which limited by acceptable privacy guarantees can be provided, under only a minor degradation of the privacy level.

\cite{haeberlen2011differential} describes several different kinds of covert-channel attacks for differential privacy frameworks.
The authors present possible countermeasures with particular focus on  one specific solution based on a new primitive called \emph{predictable transactions}.

\cite{cheu2019manipulation} presents a systematic study of a fundamental limitation of locally differentially private protocols with respect to their vulnerability to adversarial manipulation.
The authors also present a solution to provide increased security via a protocol that deploys local differential privacy and reinforces it with cryptographic techniques.

Finally, \cite{ilvento2020implementing} presents a study similar to the one first introduced in~\cite{Mir12}, but concentrate on the exponential mechanism of McSherry and Talwar~\cite{MT07}.

\section{Conclusion}\label{sec:conclusion}
In this paper we have presented an alternative defence to a particular floating point vulnerability in differential privacy.
The solution we presented, using the infinite divisibility of probability distributions, is simple to understand, quick and easy to implement, difficult to attack, and generalisable to different probability distributions and system precisions.
Of particular interest is the ability to generate samples from the Laplace distribution in a single statement of code (Section~\ref{sc:samp:lap}), with strong attack guarantees.
We believe this is an important contribution to the literature on differential privacy in mitigating some of the risks associated with operating in a floating point environment.

\begin{acks}
This work was supported by European Union's Horizon 2020 research and innovation programme under grant number 951911 -- AI4Media. The authors also wish to thank David Malone (Hamilton Institute, Maynooth University) for useful discussions at the beginning of this work.
\end{acks}

\bibliographystyle{acmalpha}
\bibliography{refs}

\appendix
\section*{Appendix}
\section{Probability density functions}\label{sc:ap:pdf}

The following probability distributions are referenced in Section~\ref{sc:main}.

\subsection{Uniform distribution}
The uniform distribution on the interval $[a, b] \subset \mathbb{R}$, $-\infty < a < b < \infty$, is given by the \ac{PDF}
$$f_{\U(a, b)}(x) = \begin{cases}
\frac{1}{b-a} & \text{if } x \in [a, b],\\
0 & \text{otherwise}.
\end{cases}$$

We make use of the uniform distribution $\U(0, 1)$ on the unit interval $[0,1]$.

\subsection{Gaussian distribution}\label{sc:ap:gau}
The Gaussian distribution with mean $\mu$ and variance $\sigma^2$ is given by the \ac{PDF}
$$f_{\N(\mu, \sigma)}(x) = \frac{1}{\sigma \sqrt{2 \pi}} e^{-\frac{1}{2} \left(\frac{x-\mu}{\sigma}\right)^2}.$$

We refer to the case when $\mu = 0$ and $\sigma = 1$ as the \emph{standard Gaussian distribution}.
If $N \sim \N(0, 1)$, then $\sigma N + \mu \sim \N(\mu, \sigma)$.

\subsection{Laplace distribution}
The Laplace distribution with mean $\mu$ and variance $2b^2$ is given by the \ac{PDF}
$$f_{\Lap(\mu, b)}(x) = \frac{1}{2b} e^{-\frac{|x - \mu|}{b}}.$$

We refer to the case when $\mu = 0$ and $b = 1$ as the \emph{standard Laplace distribution}.
If $L \sim \Lap(0, 1)$, then $b L + \mu \sim \Lap(\mu, b)$.

\subsection{Exponential distribution}
The exponential distribution with mean $\frac{1}{\lambda}$ and variance $\frac{1}{\lambda^2}$ is given by the \ac{PDF}
$$f_{\Exp(\lambda)}(x) = \lambda e^{-\lambda x}.$$

We refer to the case when $\lambda = 1$ as the \emph{standard exponential distribution}.
If $E \sim \Exp(1)$, then $\frac{E}{\lambda} \sim \Exp(\lambda)$.

\subsection{Gamma distribution}
The gamma distribution with mean $k \theta$ and variance $k \theta^2$ is given by the \ac{PDF}
$$f_{\Gamma(k, \theta)}(x) = \frac{1}{\Gamma(k) \theta^k} x^{k-1} e^{-\frac{x}{\theta}}.$$

If $G \sim \Gamma(k, \theta)$, then $c G \sim \Gamma(k, c \theta)$ for any $c > 0$.

\subsection{Chi-squared distribution}
The chi-squared distribution with $k \in \mathbb{N}$ degrees of freedom is given by the \ac{PDF}
$$f_{\chi^2(k)}(x) = \frac{1}{2 ^{\frac{k}{2}} \Gamma\left(\frac{k}{2}\right)} x^{\frac{k}{2} - 1} e^{-\frac{x}{2}}.$$

\section{Code samples}\label{sc:ab:code}

The following code samples were used in estimating execution time for different implementations.
This code was run using Python~3.8.6.

\subsection{Na\"ive sampling}
The na\"ive standard Laplace sampling given by (\ref{eq:inv:lap}) was implemented using:
\begin{lstlisting}
def laplace_naive():
    u = random()
    return copysign(log(1 - 2 * abs(u - 0.5)), u - 0.5)
\end{lstlisting}

\subsection{Theorem~\ref{th:lap} sampling}
The implementation of Theorem~\ref{th:lap} was given by:
\begin{lstlisting}
def laplace_theorem51():
    return log(1 - random()) * cos(pi * random()) + log(1 - random()) * cos(pi * random())
\end{lstlisting}

\subsection{Sampling with \texttt{math} and \texttt{random}}
We combine the Gaussian and Laplace sampling procedures from (\ref{eq:gau:div}) and (\ref{pr:lap:1}) to generate standard Laplace samples from $8n$ uniform variates using the \texttt{math} and \texttt{random} libraries as follows:
\begin{lstlisting}
def gaussian_sum(n=1):
    return sum(normalvariate(0, 1) for i in range(n))
    
def laplace_math_and_random(n=1):
    return (gaussian_sum(n) * gaussian_sum(n) - gaussian_sum(n) * gaussian_sum(n)) / n
\end{lstlisting}

\subsection{Sampling with Numpy}
Finally, we present an implementation of the same procedure using the popular Numpy package, leveraging its C-based code for faster computations with larger $n$:
\begin{lstlisting}
import numpy as np

def laplace_numpy(n=1):
    g1, g2, g3, g4 = np.random.standard_normal(size=(4, 2 * n)).sum(axis=1)
    return (g1 * g2 - g3 * g4) / 2 / n
\end{lstlisting}

\end{document}